\documentclass[letter,12pt,reqno,onside]{amsart}
\usepackage{etex}
\usepackage{graphicx}
\usepackage{setspace}
\usepackage{amsmath}
\usepackage{marginnote}
\usepackage{amsfonts}
\usepackage{amssymb}
\usepackage{pictex}
\usepackage{amsthm}
\usepackage{graphics,epsfig,verbatim,bm,latexsym,url,amsbsy}
\usepackage{rotating}
\usepackage[authoryear,round]{natbib}
\usepackage{mathrsfs}
\usepackage{multirow}
\usepackage{array}
\usepackage{url}
\usepackage{myundertilde}
\usepackage{textcomp}
\usepackage{hyperref}

\newcommand*{\mb}{\bm}
\newcommand*{\mbb}{\mathbb}
\newcommand*{\mc}{\mathcal}
\newcommand*{\la}{\langle}
\newcommand*{\ra}{\rangle}

\usepackage[top=1in, bottom=1in, left=1in, right=1in, heightrounded,
marginparwidth=2in, marginparsep=2mm]{geometry}

\usepackage{float}
\usepackage{subcaption}
\usepackage{xcolor}
\usepackage{bm}
\usepackage{tikz-cd}
\usepackage{ragged2e}
\usepackage{microtype}
\newcommand{\tr}{\mathsf{T}}

\theoremstyle{definition} 
\theoremstyle{definition} 
\theoremstyle{definition} 
\theoremstyle{definition} 
\theoremstyle{definition} 
\theoremstyle{definition} \newtheorem{definition}{Definition}
\theoremstyle{definition} 
\theoremstyle{definition} \newtheorem{lemma}{Lemma}
\theoremstyle{definition} \newtheorem{theorem}{Theorem}
\theoremstyle{definition} 
\theoremstyle{definition} 
\theoremstyle{definition}\newtheorem{proposition}{Proposition}
\theoremstyle{definition} 
\theoremstyle{definition} 
\theoremstyle{definition} \newtheorem{assumption}{Assumption}
\theoremstyle{definition} 
\theoremstyle{definition} \newtheorem{fact}{Fact}
\theoremstyle{definition} 

\theoremstyle{definition}

\long\def\symbolfootnote[#1]#2{\begingroup%
\def\thefootnote{\fnsymbol{footnote}}\footnote[#1]{#2}\endgroup}

\usepackage{footnote}

\renewcommand{\max}{\operatornamewithlimits{max}}

\newcommand{\be}{\begin{equation}}
\newcommand{\ee}{\end{equation}}
\newcommand{\bes}{\begin{equation*}}
\newcommand{\ees}{\end{equation*}}

\newcommand{\rdots}{\mathinner{%
  \mkern1mu\raise1pt\hbox{.}%
  \mkern2mu\raise4pt\hbox{.}%
  \mkern2mu\raise7pt\vbox{\kern7pt\hbox{.}}\mkern1mu}}

\allowdisplaybreaks

\begin{document}

\title[Discord and Harmony in Networks]{Discord and Harmony in Networks}

\author[]{Andrea Galeotti \and Benjamin Golub \and Sanjeev Goyal \and Rithvik Rao}

\date{\today}

\thanks{%
Joey Feffer and Zo\"{e} Hitzig provided exceptional research assistance. Andrea Galeotti gratefully acknowledges financial support from the European Research Council through the ERC-consolidator grant (award no. 724356) and the European University Institute through the Internal Research Grant. Benjamin Golub gratefully acknowledges financial support from The Pershing Square Fund for Research on the Foundations of Human Behavior and the National Science Foundation (SES-1658940, SES-1629446).
Galeotti: Department of Economics, London Business School, agaleotti@london.edu.
Golub:  Departments of Economics and Computer Science, Northwestern University, benjamin.golub@northwestern.edu.
Goyal: Faculty of Economics and Christ's College, University of Cambridge, sg472@cam.ac.uk.
Rao: School of Engineering and Applied Sciences, Harvard University, rithvikrao@college.harvard.edu}

\maketitle

\begin{abstract}
Consider a coordination game played on a network, where  agents prefer taking actions closer to those of their neighbors and to their own ideal points in action space. We explore how the welfare outcomes of a coordination game depend on network structure and the distribution of ideal points throughout the network. To this end, we imagine a benevolent or adversarial planner who intervenes, at a cost, to change ideal points in order to maximize or minimize utilitarian welfare subject to a constraint. A complete characterization of optimal interventions is obtained by decomposing interventions into principal components of the network's adjacency matrix. Welfare is most sensitive to interventions proportional to the last principal component, which focus on \emph{local} disagreement. A welfare-maximizing planner optimally works to reduce local disagreement, bringing the ideal points of neighbors closer together, whereas a malevolent adversary optimally drives neighbors' ideal points apart to decrease welfare. Such welfare-maximizing/minimizing interventions are very different from ones that would be done to change some traditional measures of discord, such as the cross-sectional variation of equilibrium actions. In fact, an adversary sowing disagreement to \emph{maximize} her impact on welfare will \emph{minimize} her impact on global variation in equilibrium actions, underscoring a tension between improving welfare and increasing global cohesion of equilibrium behavior. 
\end{abstract}
\newpage

\section{Introduction}

Consider a simple coordination game played on a network. Each player takes an action, having an incentive to bring this action closer to both a personal \emph{ideal point} and to the actions of neighbors. In the absence of any coordination concerns, each player would  set their actions equal to their ideal points; we thus also call an ideal point a \emph{favorite action}. Coordination concerns typically change this, pulling an agent's choices in equilibrium toward the ideal points of network neighbors, as well as of those farther away with whom the agent interacts only indirectly. A number of examples motivate our setup. The action may be declaring political opinions or values in a setting where it is costly to disagree with friends, but also costly to distort one's true position from the ideal point of sincere opinion. Alternatively, an action might be a choice in a technological space. For instance, in a software company, designer preferences inform tradeoffs between usability and power in the tools they use, but all are better off when their tools are more compatible with those of their colleagues.\footnote{This interpretation of actions as choices in a technological space aligns with standard models in the literature on organizations---see, e.g., \citet*{calvo2015communication}.} In this example, the network is determined by collaboration relationships, i.e.\ which designers work together.\footnote{In these examples, and throughout, we take the network to be exogenous to the decisions in question, which is often realistic in the short run. Endogenous network formation is, as always, an important concern.}

The broad question we are concerned with is how the favorite actions and the network jointly determine welfare. Given a network, how do changes in agents' ideal points affect the efficiency of equilibrium outcomes? When can relatively small changes in these ideal points have large welfare impacts? We operationalize this question by imagining a planner who can, at a cost, change favorite actions. Supposing an adversary can undertake costly influence activities and change people's views, how would she do so if her goal was to increase miscoordination? Turning to the organization example, if managers can exert influence, provide encouragement, and offer incentives to change agents' inclinations, what changes would a benevolent manager undertake to maximize welfare? By understanding what such planners would do, we can understand how the relationship between favorite points and the network determines welfare. Such insights will also be relevant for problems concerning the composition of a team; rather than directly manipulating a particular person's incentives, a planner may instead choose \emph{whom} to put in a certain organizational role or position. Such interventions require careful analysis of the welfare implications of the joint arrangement of ideal points and network links. Our results shed light on these issues.

To analyze this intervention problem, we take a spectral approach. That is, we write the relevant optimization problems in terms of functions of eigenvalues and eigenvectors of the network, which are important invariants often used to capture various aspects of network structure. Working in a ``principal component'' basis permits legible characterizations of equilibrium outcomes and optimal interventions. (In contrast, in a natural basis the solutions to our optimization problems would be unwieldy and would not shed much light on the relationship between structural features of the network and the optimal intervention.)  Our main findings are characterizations of the optimal intervention using certain eigenvectors and substantive implications for what such a planner focuses on.

Our main result, Theorem \ref{thm:main}, is that the most welfare-consequential changes in favorite actions focus primarily (in a sense we make precise) on the \emph{last} eigenvector of the network: the one associated with its lowest (typically most negative) eigenvalue. Beyond this, there is a monotonicity to the structure of interventions: principal components with lower eigenvalues receive less focus in optimal interventions. In special cases that we describe, the focus on the lowest principal component can be exclusive: at the optimal intervention, all disagreement in favorite actions is loaded onto this one principal component.\footnote{We will use \textit{principal components} and \textit{eigenvectors} interchangeably.} Our results also imply that explicit functions of certain eigenvalues can summarize the sensitivity of equilibrium welfare to optimal perturbations of ideal points. This gives an answer to the question posed at the beginning about how sensitive welfare is to the configuration of ideal points.

Going beyond a characterization in terms of a canonical graph statistic, we interpret the implications in terms of more intuitive aspects of graph structure. A key distinction we emphasize is between \emph{local discord}---creating disagreement\footnote{Throughout, we use ``disagreement'' to refer to differences in actions across the network.} at the ``street level,'' between neighbors---and \emph{global discord}---which creates disagreement between separate regions. Our result implies that optimal---i.e., welfare-maximizing or minimizing---interventions have a very local focus in a precise sense. An adversary seeks to amplify disagreement between neighbors, pushing neighbors' favorite points apart.

Notably, the interventions that best achieve this are quite distinct from those that best create global discord in the network. Indeed, creating global discord is in tension with reducing welfare. When an adversary optimally sows discord in ideal points to reduce welfare, this leads to a low level of variation across the population in equilibrium behavior, in a sense we make precise. Relatedly, if there is a certain amount of diversity (cross-sectional variation) in ideal points, it turns out that agents are best off when they are arranged so that they agree with their immediate neighbors and disagree with those distant from them in the network. This naturally leads to societies sustaining more diversity in equilibrium behavior and appearing more divided in a global sense. To summarize, our main results deliver stark predictions about which aspects of the configuration of ideal points matter for welfare, and these are quite different from what we might expect from standard intuitions about discord (as we elaborate on in our discussion of related literature below). 

Finally, a conceptual point in our analysis is that intervention problems can be useful metaphors for understanding what structural features matter for welfare in a game. In some cases, a planner may indeed be intervening quite explicitly. For instance, an adversary may be seeking to use social media to sow division in opinions and cause costly tensions between neighbors.\footnote{See \citet{House} for a report on such activities.} But in many other problems, an analyst may simply be interested in understanding which shifts in exogenous primitives most affect welfare; hypothetical intervention problems shed light on this even when an intervention is not literally being designed.

\subsection{Related work}

Broadly, we are situated in the economics literature on network games, surveyed, for example by \citet*{JZ14}; see also the bibliography of \citet*{BKD14}.  This literature, in terms of techniques and many of the measures that are relevant, is also related to the literature on opinion updating and social learning in networks, going back to \citet{DeGroot74} and surveyed by \citet*{acemoglu2011opinion} and \citet*{GS16}.

Within this broad literature, our project is distinguished by two aspects of our substantive focus. First, we are interested in a welfare objective. While most works in the economics literature on network games of course touch on efficiency and welfare considerations, the main outcome of interest is often an overall level of activity or knowledge---as, for instance, in \citet*{BCAZ06} and \citet*{KKT15}.\footnote{Spectral methods play a significant role in the study of global influence, which is closely connected to the Perron vector (\emph{eigenvector centrality}), as in \citet*{BCAZ06, ACOTS12}. Different eigenvectors matter in our analysis because we are not concerned with first moments of behavior but rather variation and disagreement across agents.} There are fewer that are focused on social welfare. An early contribution, with a price of anarchy approach, is \citet*{BKO11}, who give bounds on the welfare difference between equilibrium and a social optimum under the \citet*{FJ99} social learning model. Another closely related contribution is due to \citet*{AP07}, who study fundamental structural properties of equilibrium welfare in beauty contests among other classes of games. In macroeconomics, the welfare implications of shocks are studied by \citet{baqaee2019macroeconomic} and \citet{baqaee2020productivity}. \citet*{GGG20} and \citet*{AK20} are perhaps the closest in that they consider  welfare-optimal interventions.\footnote{Targeting of interventions more broadly is studied, e.g., in \citet*{AJB00, Valente12, KKT15}.} However, the class of games considered is very different: investment or public goods games. These involve quite different externalities from the ones that are relevant for coordination games and discord, which is what we focus on (as discussed by \citet*{AP07}).\footnote{\citet*{BKD14} focus on stability of equilibrium rather than targeting, but find that eigenvectors related to the ones we study matter in public goods games.}

Issues of miscoordination and discord are touched on in another thread of literature. This work analyzes how the configuration of agents' attributes (initial opinions, ideal points, etc.) affects the dynamics and ultimate outcomes of processes in social networks. The connection of these outcomes to spectral aspects of the network was noted by \citet*{DVZ03}, and further developed by \citet{GJ12}, which highlighted the relation to spectral clustering.\footnote{For more on various segregation measures that come up in various related contexts, see \citet*{Morris00, DVZ03, CJP09, GJ12, ST07}.} An important recent contribution on discord is \citet*{GKT20}, which studies maximizing and minimizing particular measures of discord in \citet{FJ99} updating processes (which, mathematically, are closely related to our games). Crucially, in all these projects, the notion of discord that is of interest is a particular, exogenously given measure, rather than welfare in the game. Criteria of interest include the duration of disagreement in an updating process, average disagreement across individuals in the network, etc. In our work, we get the objective from the preferences of the players themselves, maximizing utilitarian welfare. Thus, while the principal component approach overlaps methodologically with many of these studies, the welfare-oriented questions we ask lead to insights quite different from those in the prior literature. Indeed, a theme in the prior literature is that \emph{global} discord between loosely connected regions is most important in slowing down agreement \citep*{DVZ03,GJ12}. The component of disagreement that most strongly remains after a long period of updating opinions is proportional to the second eigenvector. As we will show, our results deliver a starkly different message. The classical spectral cut component---the second eigenvector that partitions the network into pieces that are relatively loosely connected to each other---is the \emph{least} consequential for welfare in our setting. \citet*{GKT20} has more subtle results showing that there is no clear ordering of how an adversary focuses effort on various spectral components of disagreement. This is natural in view of their wider class of objectives. We show that for standard welfare-oriented objectives in coordination games, there is a clear ordering, with the last eigenvector being of primary importance. Finally, our Theorem \ref{thm:main} imposes less structure on the class of possible interventions than, e.g., \citet*{GJ12} or \citet*{GKT20}; we allow perturbations around an arbitrary status quo and, for small interventions, can deal with a large class of intervention cost functions.

\section{Model, Basic Facts, and Definitions}

In this section, we state the model and definitions we need. We also mention some standard results on the structure of equilibrium that serve as a foundation for our subsequent results.

\subsection{Coordination game}

We consider a one-shot game played between individuals $\mc{N} = \{1 , \dots, n \}$, with a typical individual denoted $i$. Each individual takes an \emph{action}\footnote{The one-dimensional space is for simplicity: our analysis extends without much change to actions in an arbitrary Euclidean space.} $a_i \in \mbb{R}$. We are given a \emph{favorite action} $f_i \in \mbb{R}$ for each agent and a network with a weighted adjacency matrix $\mb{G} \in \mbb{R}^{n \times n}$. An agent's payoff is determined by her favorite action and the actions of her neighbors in $\bm{G}$. We write the vector of actions as $\mb{a} \in \mbb{R}^n$, and the vector of favorite actions as $\mb{f} \in \mbb{R}^n$. Individual $i$ chooses $a_i$, while $\mb{f}$ and $\mb{G}$ are exogenous.

We will assume that $\bm{G}$ is row-stochastic and symmetric, and that each $i$ meets and interacts with $j$ with probability $g_{ij}$. The payoff to an agent $i$ of interacting with agent $j$ is given by:
\begin{equation}\label{eqn:payoff}
    v_i (a_i, a_j) = - \underbrace{\beta(a_i - a_j)^2}_{\textrm{miscoordination}} - \underbrace{(1 - \beta) (a_i - f_i)^2}_{\textrm{distance from favorite action}},
\end{equation}
where $\beta \in [0, 1)$ determines the relative payoff weight of miscoordination with other individuals and distance from an individual's favorite action. The expected payoff of individual $i$ given action profile $\mb{a}$ is
\begin{equation*}
   V_i(\bm{a}) = \sum_j g_{ij} v_i (a_i, a_j).
\end{equation*} Utilitarian welfare is defined by $$ V(\bm{a}) = \sum_i V_i(\bm{a}). $$

\subsection{Nash equilibrium: A formula and a few basic properties}

Here we review a few standard facts about the Nash equilibrium.

Fixing $\mb{f}$ and $\mb{G}$, the first-order condition characterizing the Nash equilibrium action profile is given by
\begin{equation*}
    a_i^* = \beta \sum_j g_{ij} a_j^* + (1-\beta) f_i,
\end{equation*}
and this can be rewritten in vector notation to show that any Nash equilibrium action profile $\mb{a}^*$ must satisfy
\begin{equation}\label{eqn:nashEqm}
    \mb{a^*} = (1-\beta)(\mb{I} - \beta\mb{G})^{-1} \mb{f}.
\end{equation}

We make the following two assumptions, the first of which has already been mentioned above.
\begin{assumption}\label{asm:symmetric}
    The adjacency matrix $\mb{G}$ is row-stochastic and symmetric.
\end{assumption}

Assumption \ref{asm:symmetric} is implied by the description of $\bm{G}$ as meeting probabilities. It implies that the largest eigenvalue of $\mb{G}$ is 1 and ensures that (\ref{eqn:nashEqm}) characterizes a unique, stable Nash equilibrium \citep*{BCAZ06,BKD14}. Indeed, we have the following fact:
\begin{fact}  The game has a unique Nash equilibrium, which is in pure strategies and given by (\ref{eqn:nashEqm}). In this equilibrium, each $a_i$ is a (possibly different) weighted average of the $f_j$. \label{fact:Nash}
\end{fact}
\begin{proof} It is straightforward to check that the second-order conditions for optimization hold, so the first-order condition is necessary and sufficient. Assumption \ref{asm:symmetric} ensures $\beta\bm{G}$ has spectral radius less than $1$ and so we may rewrite (\ref{eqn:nashEqm}) by the Neumann series as \begin{equation} \label{eq:Neumann}\bm{a}^* = \underbrace{\left(\sum_{t=0}^\infty (1-\beta){\beta^t} \bm{G}^t \right)}_{\bm{W}} \bm{f}.\end{equation} Letting $\bm{W}$ be the matrix in parentheses, we see that it is a weighted average (with weights $(1-\beta){\beta^t}$) of stochastic matrices $\bm{G}^t$, so $\bm{W}$ is itself stochastic. Thus, $a_i=\bm{W}_{i \bullet} \bm{f}$, where $\bm{W}_{i \bullet}$ is row $i$ of $\bm{W}$. \end{proof}

\begin{figure}
    \centering
    \includegraphics[width=0.8\textwidth]{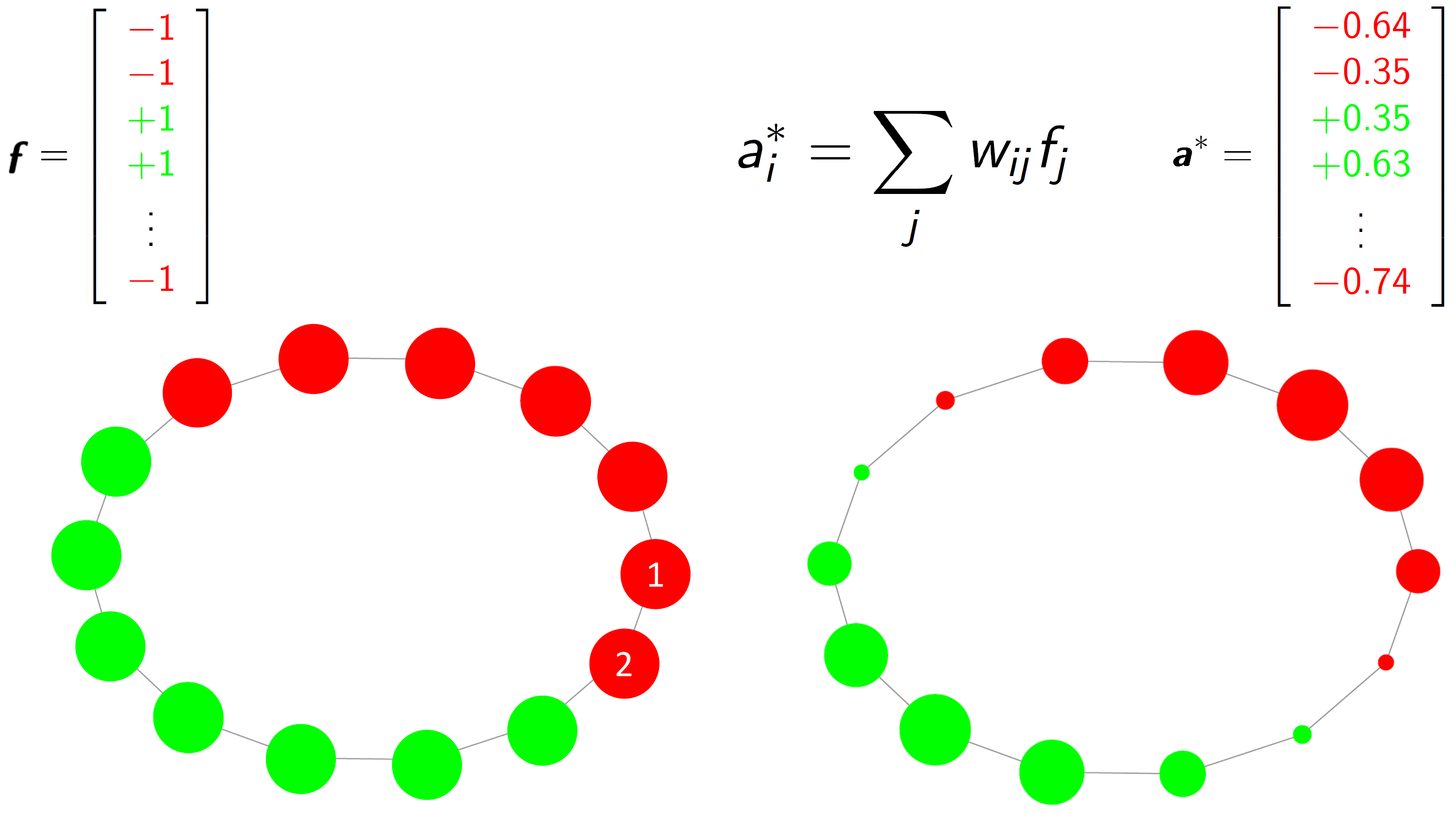}
    \caption{An illustration of equilibrium for a given network (the circle,  where $g_{ij}=0.5$ whenever $i$ and $j$ are adjacent). The node labels are as shown, continuing clockwise.  On the left we depict a particular vector $\bm{f}$. When we depict a node-indexed vector (such as $\bm{f}$ or $\bm{a}$) visually, our convention is that entries with positive sign are indicated by making the node green, while entries with negative sign are indicated by making the node red. The size of each node corresponds to the magnitude of its entry. On the left side we write and illustrate $\bm{f}$, while on the right side we calculate $\bm{a}^*$ using Fact \ref{fact:Nash} and illustrate it in the same type of diagram.}
    \label{fig:example-eqm}
\end{figure}

To illustrate the implications of Fact \ref{fact:Nash}, consider Figure \ref{fig:example-eqm}. There, we take a particular vector $\bm{f}$ where half of the agents (those in the bottom left) have a favorite action of $+1$, while those in the top right have a favorite action of $-1$. We calculate equilibrium using (\ref{eqn:nashEqm}) for a particular value of $\beta$. We can then see the structure of equilibrium asserted in Fact \ref{fact:Nash}: everyone's action is a weighted average of $+1$'s and $-1$'s, with closer agents weighted more and farther agents weighted less.\footnote{Note also from the form of (\ref{eq:Neumann}) that the Nash equilibrium can be seen as the average of $\bm{G}^t \bm{f}$ for $t\in\{0,1,2,\ldots\}$, which are the outcomes of \citet{DeGroot74} or \citet{FJ99} learning or myopic updating at various times $t$; see \citet{GJ12}. This explains the close connection between properties of equilibria in network games and the dynamics of certain updating/learning processes in networks; see also \citet*{GKT20}.}

While the favorite actions exhibit a very stark difference between groups, the equilibrium actions \emph{attenuate} the diversity of favorite actions. ``Boundary'' agents average together roughly as many $+1$'s as $-1$'s, and end up with equilibrium actions close to $0$. They are quite far from their favorite actions, though they are fairly closely coordinated with their neighbors. Agents deep in the bottom left or top right end up with actions that are much more extreme, and therefore closer to their ideal points. Even they, however, end up with actions less extreme than the extremes of $\bm{f}$, illustrating the attenuation property of best responses.\footnote{This can be seen by noting from (\ref{eq:Neumann}) that in a connected graph, each agent puts positive weight on all others, and thus even the most extreme agents become less extreme. The higher $\beta$ is, the stronger the attenuation.}

We will make a final, technical, assumption to simplify the statement of some results. This holds generically (over the choice of weights in the symmetric matrix).
\begin{assumption}\label{asm:distinct}
    All eigenvalues of $\mb{G}$ are distinct. 
\end{assumption}

\subsection{Planner interventions and objective}

Our main interest is in understanding how, in examples such as the one just discussed, welfare is affected by changes in the favorite actions of various players.
We investigate this question by considering a planner who can modify the vector of favorite actions: the favorite actions $\hat{\mb{f}}$ are modified by some perturbation vector $\bm{\delta} \in \mbb{R}^n$. Formally, the planner's problem is given by
\begin{equation}\label{eqn:genPlanner}
    \begin{aligned}
    \max_{\bm{\delta}} \quad & \gamma V(\bm{a}^*)\\
    \textrm{s.t.} \quad & \mb{f} = \hat{\mb{f}} + \bm{\delta}\\
    & \mb{a^*} = (1-\beta)(\mb{I} - \beta\mb{G})^{-1} \mb{f},\\
    & c(\bm{\delta}) \leq C.
\end{aligned}
\end{equation}

The parameter $\gamma$ scaling the objective is $+1$ or $-1$, corresponding to the planner being benevolent or malevolent, respectively. The constraint $c(\bm{\delta}) \leq C$ limits the feasible interventions. The cost function $c(\cdot)$ is for now taken to be arbitrary. For various results, we will give specific cost functions: for example, constraining interventions to a ball of fixed size around the status quo. In our most general results in Section \ref{subsec:cost}, we study classes of cost functions satisfying certain assumptions, such as that interventions have (at least locally) convex costs. The number $C\geq0$ is called the \emph{budget}.

\subsection{Principal components: Definitions and notation}

We introduce notation for the key objects that play a role in our approach: the principal components of the network $\mb{G}$. We write the spectral decomposition of $\mb{G}$ as follows:
\begin{equation}\label{eqn:eigendecomp}
    \bm{G}=\underbrace{\left[\begin{array}{ccc}
    \mid & & \mid \\
    \mb{u}^{1} & \ldots & \mb{u}^{n} \\
    \mid & & \mid
    \end{array}\right]}_{\mb{U}: \text { eigenvectors }} ~~\underbrace{\left[\begin{array}{ccc}
    \lambda_{1} & & 0 \\
     & \ddots & \\
    0 & &  \lambda_{n}
    \end{array}\right]}_{\mb{\Lambda}: \text { eigenvalues }} ~~\underbrace{\left[\begin{array}{ccc}
    \text{\textemdash} & \left(\mb{u}^{1}\right)^{\intercal} & \text{\textemdash} \\
     & \vdots & \\
    \text{\textemdash} & \left(\mb{u}^{n}\right)^{\intercal} & \text{\textemdash}
    \end{array}\right]}_{\mb{U}^{\intercal}: \text { eigenvectors }}.
\end{equation}

\begin{figure}
    \centering
    \includegraphics[width=0.8\textwidth]{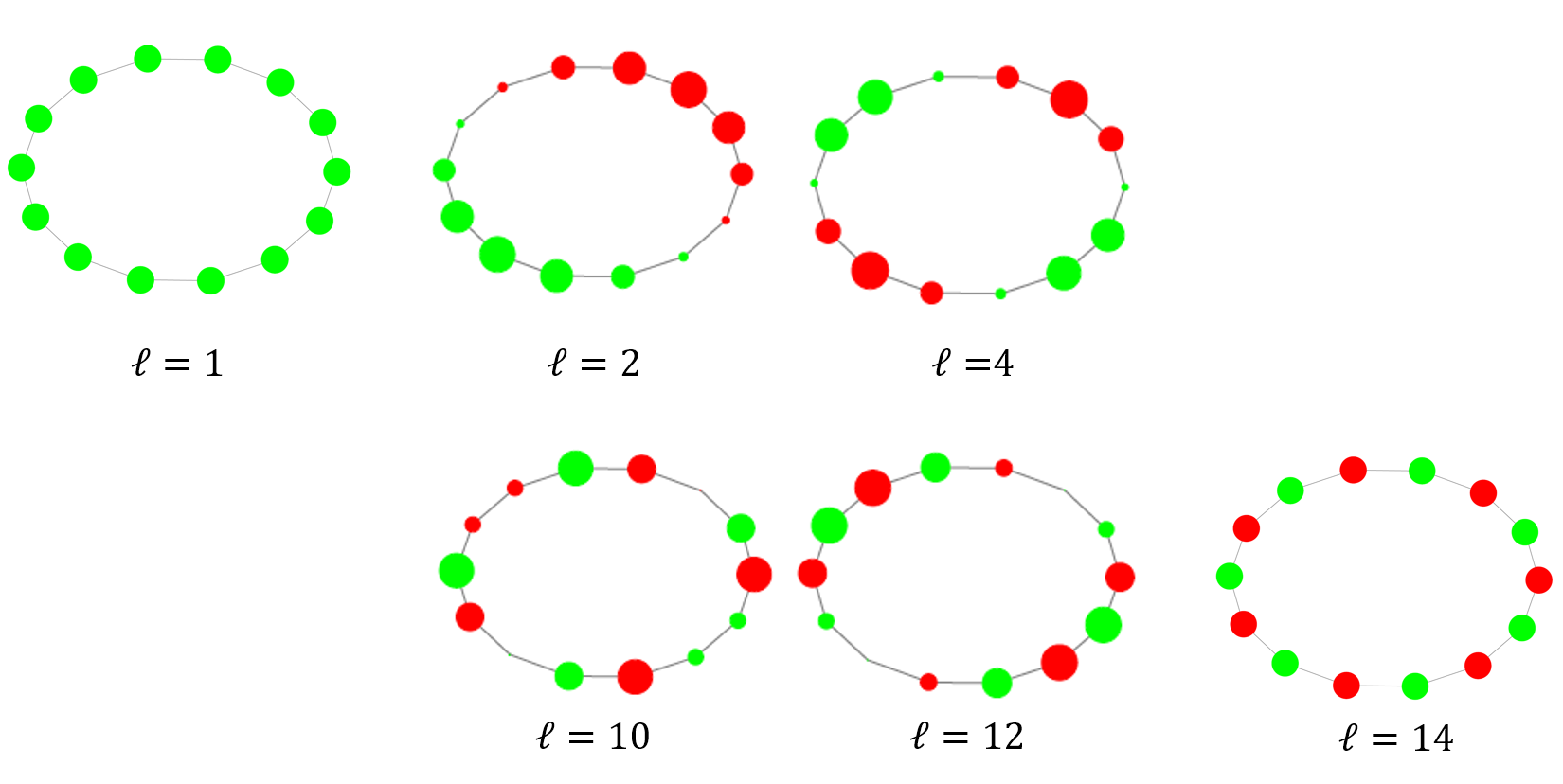}
    \caption{Six eigenvectors of a circle network. The eigenvector $\bm{u}^\ell$ corresponding to the $\ell^{\text{th}}$-highest eigenvalue $\lambda_\ell$, is depicted using the same visual convention we introduced in Figure \ref{fig:example-eqm}. Note that the eigenvectors higher eigenvalues (higher $\ell$) vary ``more slowly'' over the circle than those with lower eigenvalues (higher $\ell$).}
    \label{fig:spectrum}
\end{figure}

Here, $\mb{U}$ gives an {orthonormal basis} of eigenvectors. We adopt the convention that the eigenvectors and eigenvalues are arranged so that $\lambda_1 \geq \lambda_2 \geq \dots \geq \lambda_n $. We will refer to the eigenvector corresponding to $\lambda_{\ell}$ as $\bm{u}^{\ell}$. For any vector $\bm{z} \in \mathbb{R}^n$, let $\underline{\bm{z}}=\bm{U}^\tr \bm{z}$. We will refer to $\underline{z}_\ell$ as the projection of $\bm{z}$ onto the $\ell^{\text{th}}$ principal component, or the magnitude of $\bm{z}$ in that component.

Figure \ref{fig:spectrum} illustrates some principal components of an example network.

Throughout, we use $\langle \bm{y}, \bm{z} \rangle = \sum_{i \in \mathcal{N}} y_i z_i$ to denote the Euclidean dot product, and we let $\Vert \bm{z} \Vert = \langle \bm{z}, \bm{z} \rangle^{1/2}$ denote the Euclidean norm. Since the eigenvectors are normalized, they satisfy $\Vert \bm{z} \Vert = 1$.

\section{Two Simple Planner Problems and Two Distinguished Principal Components}

Certain principal components will play an important role in our analysis. For instance, the \emph{last} principal component, the eigenvector $\bm{u}^n$ corresponding to the lowest eigenvalue $\lambda_n$ will correspond to the direction in which interventions are most consequential for welfare. It will also be helpful to contrast it with another eigenvector, $\bm{u}^2$, the one that corresponds to the second-highest eigenvalue $\lambda_2$. This eigenvector, which has been important in prior studies of segregation and homophily \citep*{DVZ03,GJ12}, turns out to describe \emph{least} welfare-consequential interventions, and so it will serve as an important foil or contrast for some of our results.

To show the role these eigenvectors play in optimization problems, we define a special case of the planner's problem, in which the planner chooses any $\mb{f}$ on a sphere of radius 1 to maximize or minimize welfare. This corresponds to holding the cross-sectional variation of favorite actions fixed, and distributing a ``fixed'' amount of disagreement to achieve the objective. In this section, we dispense with $\bm{\delta}$ and work with choosing the vector $\bm{f}$ directly, since the simplicity of the problem makes this change straightforward. Thus, we can simply consider how the planner decides to allocate disagreement in her choice of $\bm{f}$, subject to a norm constraint.

The optimization problem of interest is defined by
\begin{equation} \label{eq:simpleplanner}
    \begin{aligned}
    \max_{\mb{f}} \quad & \gamma V(\bm{a}^*)\\
    \textrm{s.t.} \quad  & \mb{a^*} = (1-\beta)(\mb{I} - \beta\mb{G})^{-1} \mb{f} \\
  & \left\Vert \mb{f} \right\Vert = 1. \\
\end{aligned}
\end{equation}

\begin{proposition}\label{prop:toyResults}
    Fixing $\beta$, there is an increasing function $\zeta: \mbb{R} \to \mbb{R}$ such that:
    \begin{enumerate}
        \item The optimum of (\ref{eq:simpleplanner}) for the malevolent planner ($\gamma=-1$) is achieved by $\mb{f}^* = \mb{u}^n$ and is equal to $\zeta(\lambda_n)$.
        \item The optimum of (\ref{eq:simpleplanner}) for the benevolent planner ($\gamma=1$) is achieved by $\mb{f}^* = \mb{u}^2$ and is equal to $\zeta(\lambda_2)$.
    \end{enumerate}
\end{proposition}
\begin{proof}
     We begin by writing the formula for equilibrium welfare in terms of an inner product expression depending on $\mb{f}$ and $\mb{G}$.
    \begin{align*}
        V^* &=  -\sum_i \left( (1-\beta)(a_i^* - f_i)^2 + \sum_j g_{ij} \beta(a_i^* - a_j^*)^2 \right)\\
        &= -\la \mb{a}^*, ((1+\beta)\mb{I} - 2\beta\mb{G}) \mb{a}^* \ra + (1 - \beta)\la \mb{f} - 2\mb{a}^*, \mb{f} \ra\\
        &= -(1-\beta) \left[\la \mb{f}, \mb{f} \ra + (1- \beta)\la (\mb{I} - \beta \mb{G})^{-1} \mb{f}, ((1+\beta)\mb{I} - 2\beta \mb{G})(\mb{I} - \beta \mb{G})^{-1} \mb{f} - 2 \mb{f} \ra \right]
    \end{align*}
    
We now switch into the basis of principal components.  Recall $\underline{\mb{z}} = \mb{U}^{\intercal} \mb{z}$. Then
    $$\mb{a} = (1 - \beta) ( \mb{I} - \beta \mb{G})^{-1} \mb{f}$$
    if and only if 
    $$ \underline{\mb{a}} = (1 - \beta) (\mb{I} - \beta \mb{\Lambda})^{-1} \underline{\mb{f}}.$$ Moreover, we may replace all vectors and matrices in the above expression for $-W^*$ by their versions in the new basis. All matrices involved are diagonal, so this greatly simplifies the expression; indeed, as shown in Lemma \ref{lem:apdx:funcs} in the appendix, this yields the following expression
    \begin{align*}
        V^* = \sum_{\ell=1}^n \zeta (\lambda_\ell) \underline{f}_{\ell}^2,
    \end{align*}
    for some increasing, nonnegative function $\zeta(\lambda)$, with $\zeta(1)=0$ (so that the $\lambda_1$ term drops out, since $\lambda_1=1$). Note also that because the change of basis is orthonormal, the constraint set for $\bm{f}$ does not change.
    
    Because $\zeta$ is increasing in $\lambda$, the optimum for $\gamma=-1$ is achieved by $\mb{f}^* = \mb{u}^n$ and is equal to $\zeta(\lambda_n)$. The optimum for $\gamma=1$ is achieved by $\mb{f}^* = \mb{u}^2$ and is equal to $\zeta(\lambda_2)$.
\end{proof}

Proposition \ref{prop:toyResults} shows that when $ \mb{f} $ is constrained to a sphere, extremal welfare in the minimization problem depends on $\mb{G}$ only through an extreme eigenvalue, $\lambda_n$ or $\lambda_2$. Indeed, it remains true if we replace the constraint by  $\left\Vert \mb{f} \right\Vert \leq C$, for $C >0$, as long as we make the adjustment that $\zeta(\cdot)$ is replaced by $C\zeta(\cdot)$. Thus, $\zeta(\lambda_n)$ captures the sensitivity of welfare to the size of the invervention when the intervention is chosen optimally.\footnote{We reproduce the function $\zeta$ here for convenience, from Lemma \ref{lem:apdx:funcs}:  $$\zeta ( \lambda) = -\beta (1 - \beta) \frac{(1 - \lambda) [2 - \beta(1 +\lambda)]}{(1 - \beta \lambda)^2}.$$}

In terms of the form of intervention, loading all the diversity in favorite actions onto the last principal component is the most effective way of reducing welfare subject to an upper bound on the norm of the favorite actions. This is the first manifestation of the idea that the last principal component is the one to which welfare is most sensitive. 

In contrast, the second part of the result highlights that welfare is, in a sense, \emph{least} sensitive to disagreement along the second principal component. For fixed norm of $\bm{f}$, if we load all disagreement onto $\bm{u}^2$, welfare turns out to be the \emph{least} negative---least changed from a baseline of 0 when there is no disagreement.

\begin{figure}
    \centering
    \includegraphics[width=0.8\textwidth]{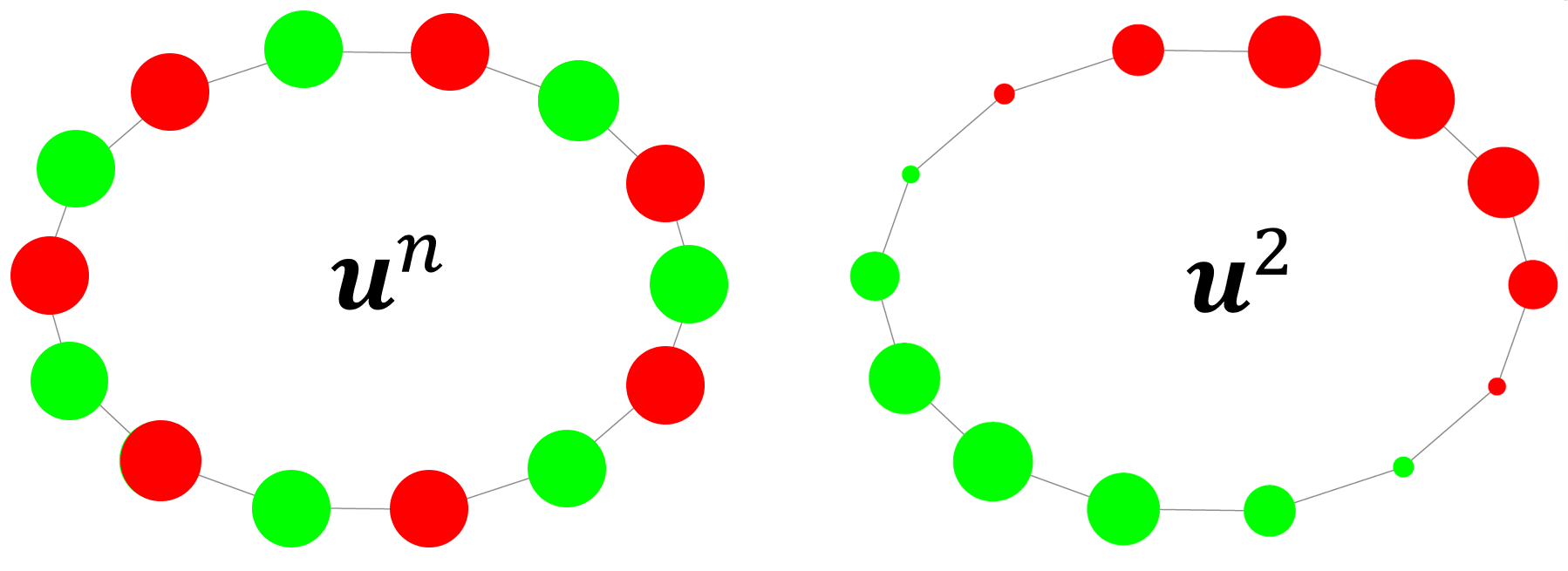}
    \caption{The eigenvectors of a circle network corresponding to the $n^{\text{th}}$ largest and $2^{\text{nd}}$ largest eigenvalues, respectively. The former maximizes local heterogeneity and separates neighbors, while the latter finds a global cut.}
    \label{fig:un_u2_comparison}
\end{figure}

Finally, it is worth remarking on the fact that $\bm{u}^1$ plays no role in the characterization. Note that in this problem $\bm{u}^1$ is a constant vector, because $\mb{G}$ is row-stochastic.\footnote{In general, this vector gives the agents' \emph{eigenvector centralities} in the network, which measure the global influence of each agent. Because of the symmetry of interactions and the fact that each agent has the same endowment of total interaction probability, there is no heterogeneity in this, but our analysis can be extended to settings where there is heterogeneity in interaction quantity.} Thus, changes in $\underline{f}_1$ correspond to constant shifts in favorite actions, which, by Fact \ref{fact:Nash} translate into the same constant shifts in equilibrium actions. These shifts do not affect welfare, and so are never used by the planner.

\subsection{Local and global disagreement at the optima}

Next, we are interested in describing how the eigenvectors identified in Proposition \ref{prop:toyResults} relate to the network, and what qualitative comparisons we can make between equilibrium behavior at the two configurations analyzed. We will show that, in a suitable sense, the last eigenvector $\bm{u}^n$ is the one that maximizes local disagreement, while the second eigenvector $\bm{u}^2$ maximizes global disagreement, subject to a constraint on norm.

We make a few definitions. Let $\mathcal{D}_{\text{R}}$ be the uniform distribution on the set  $\{(i,j) \in \mc{N} \times \mc{N} \text{ s.t. } i \neq j\}$. This corresponds to drawing a random pair. Let $\mathcal{D}_{\bm{G}}$ be the distribution on the same set obtained by drawing the pair $(i,j)$ with probability $g_{ij}/n$.  

\begin{definition} Fix a vector $\bm{z} \in \mathbb{R}^n$ such that $\sum_{i\in \mathcal{N}} z_i = 0$. 
    \begin{enumerate}
        \item The \emph{covariance of a random pair} for $\bm{z}$ is defined to be $\mathbb{E}_{(i,j) \sim \mathcal{D}_{\text{R}}} [z_i z_j ]$
        \item The \emph{covariance of neighbors} for $\bm{z}$ is defined to be $\mathbb{E}_{(i,j) \sim \mathcal{D}_{\bm{G}}} [z_i z_j ]$.
    \end{enumerate}
\end{definition}

Now we use the covariance of the actions of a pair of neighbors selected  at random\footnote{According to the same distribution that selects partners to play the bilateral game in our model.} as a measure of local disagreement, and the covariance of the actions of a random pair of agents as a measure of global disagreement. In each case, the more negative the number, the more disagreement there is of the relevant kind.

\begin{proposition}\label{prop:disagreementParadox}
    Let $\mathcal{F}$ be the set of vectors $\bm{f}$ satisfying $\sum_{i \in \mathcal{N}} f_i=0$ and $\Vert \bm{f} \Vert = 1$. The values of $\bm{f}$ in this set that maximize and minimize each quantity below are given by the following table:

    \begin{center}
        \begin{tabular}{ccc} 
        Statistic for eq'm actions $\mb{a}^{*}(\mb{f})$ & maximizer & minimizer  \\
        \hline \hline covariance of neighbors & $\mb{u}^{2}$ & $\mb{u}^{n}$ \\
        \hline covariance of random pair & $\mb{u}^{n}$ & $\mb{u}^{2}$ \\
        \end{tabular}
    \end{center}
\end{proposition}
\begin{proof}
 We show each covariance result separately.
    
  The covariance of neighbors for equilibrium actions $\bm{a}^*(\bm{f})$ is given by
    $${1 \over n} \left( \sum_{i,j \in \mathcal{N}} g_{ij} a_i^* a_j^* \right) = {1 \over n} \la \mb{a}^*, \mb{Ga}^* \ra,$$
    because we sample an agent $i$ uniformly at random from $\mathcal{N}$, and a second agent $j$ incident to $i$ (the probability of an agent $k$ being sampled is $g_{ik}$). As with Proposition \ref{prop:toyResults}, we can rewrite this expression in the principal component basis as 
    $$\sum_{\ell=1}^n \eta(\lambda_{\ell}) \underline{f}_{\ell}^2,$$
    for an increasing function $\eta(\lambda)$. (This is the content of Lemma \ref{lem:apdx:funcs} in the appendix.) Because each summand is increasing in $\lambda_{\ell}$, this expression achieves its minimum at $\mb{f}^* = \mb{u}^n$ and its maximum at $\mb{f}^* = \mb{u}^2$.
    
    The covariance of a random pair for equilibrium actions $\bm{a}^*(\bm{f})$ is given by
    $${1 \over n^2} \left(\sum_{i,j \in \mathcal{N}} a_i^* a_j^* - \sum_{i \in \mathcal{N}} (a_i^*)^2 \right),$$
    because we sample an agent $i$ uniformly at random from $\mathcal{N}$, and we sample a second agent uniformly at random from $\mathcal{N} \setminus \{i \}$. Because $\mb{G}$ is row-stochastic, its Perron vector is the all-ones vector, with eigenvalue 1. Thus the projection operator onto the eigenspace associated with eigenvalue $\lambda_1=1$ is $\mb{P}_{(1)} = \mb{1} \mb{1}^{\intercal}$. We can then rewrite the above expression as
    $$\la \mb{a}^*, \mb{P}_{(1)} \mb{a}^* \ra - \la \mb{a}^*, \mb{a}^* \ra = \la \mb{a}^*, (\mb{P}_{(1)} - \mb{I}) \mb{a}^* \ra.$$
    
  The average equilibrium action is a constant times the average of $\bm{f}$. Thus, $\mb{P}_{(1)} \mb{a}^* = \mb{0}$. It follows that the covariance-minimizing $\mb{a}$ maximizes $\la \mb{a}, \mb{a} \ra$. This expression can be written in the principal component basis as
    $$\sum_{\ell=1}^n \nu(\lambda_{\ell}) \underline{f}_{\ell}^2,$$
    for a decreasing function $\nu(\lambda)$. (See Lemma \ref{lem:apdx:funcs} for the explicit function.) Because each summand is decreasing in $\lambda_{\ell}$, this expression achieves its minimum by $\mb{f}^* = \mb{u}^n$ and its maximum by $\mb{f}^* = \mb{u}^2$.
    
    Note that $\mb{u}^1$ does not optimize either function for the same reason as mentioned previously: $\mb{u}^1$ is a constant vector for a row-stochastic $\mb{G}$, and the attenuation process keeps it constant, so any such intervention has no effect on welfare.
\end{proof}

\begin{figure}
    \centering
    \includegraphics[width=0.8\textwidth]{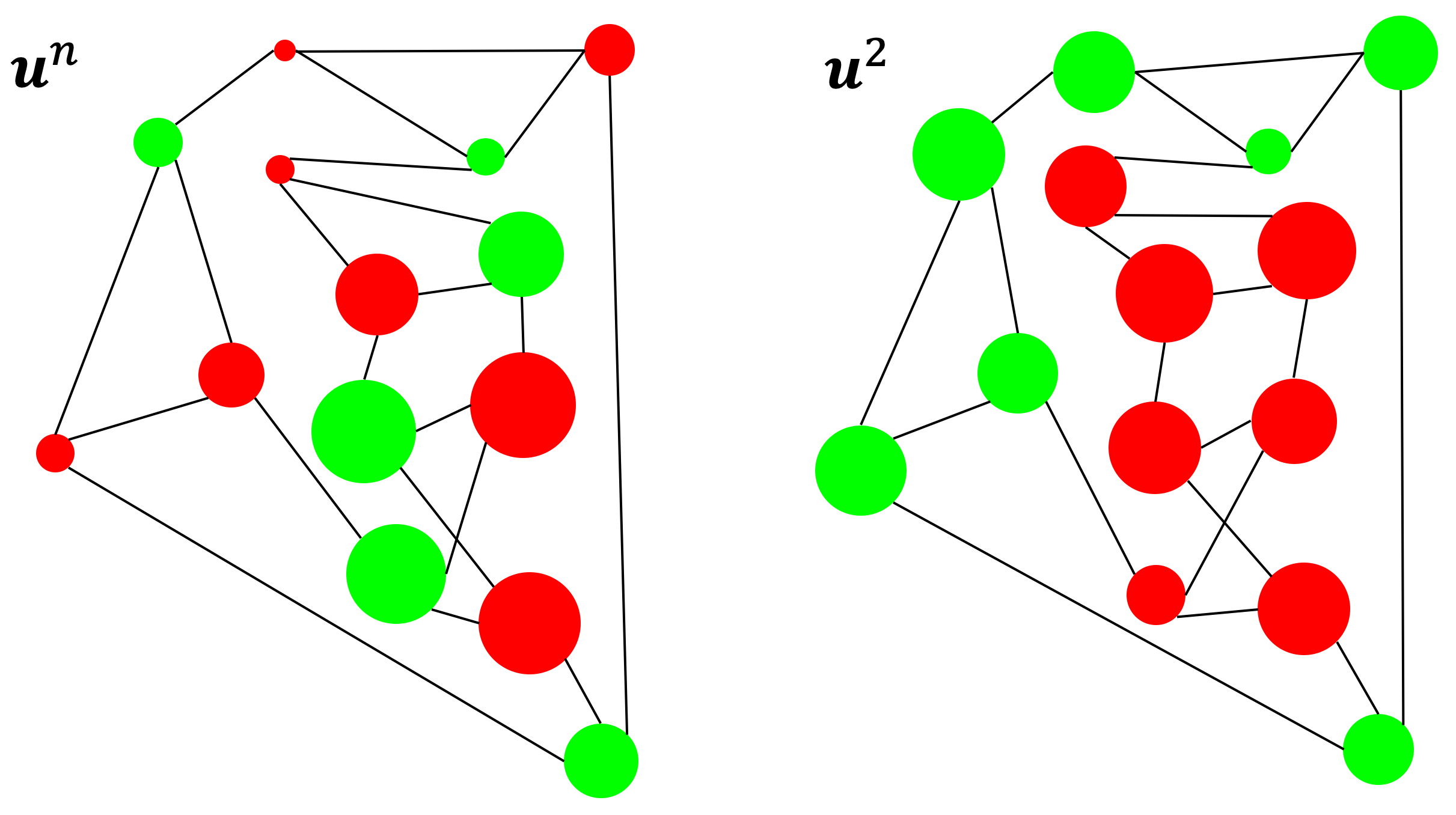}
    \caption{The $2^{\text{nd}}$ and $n^{\text{th}}$ eigenvectors of a more complex network. In this setting, $\mb{u}^n$ and $\mb{u}^2$ still respectively recover local and global structure. In particular, $\mb{u}^n$ separates neighbors by creating heterogeneity between actions (both in sign and in magnitude), while $\mb{u}^2$ globally creates two nearly-homogenous groups.
    }
    \label{fig:un_u2_realnetwork}
\end{figure}

Proposition \ref{prop:disagreementParadox} can be seen as a \emph{local-global disagreement tradeoff}: the $\bm{f}$ that maximizes local disagreement in equilibrium also minimizes global disagreement.\footnote{If $\beta$ is not too small, we can obtain equivalent results defining disagreement as the expectation of $(a_i-a_j)^2$ under the appropriate distribution of $i,j$ (random pair or neighbor).}  

\subsection{Interpretation and discussion} \label{sec:interp}

Imagine that an adversary manipulates individual ideal points in a community to reduce its members' welfare. Naively, one might expect that the consequences of this adversary's activity would be to cause global discord: to make it likely that two randomly chosen individuals would disagree strongly. Our results show that, in fact, for a given amount of cross-sectional variation in favorite points, the adversary in a sense accomplishes the opposite. We now explore this somewhat counterintuitive phenomenon.

First, let us formalize what we said in the previous paragraph. Proposition \ref{prop:toyResults} says that the malevolent planner chooses  $\bm{f} = \bm{u}^n$. Proposition \ref{prop:disagreementParadox} says that this choice causes the \emph{least} global disagreement: it creates the highest possible covariance of equilibrium actions between random pairs of individuals.

To understand the forces behind these results, we again consider the example network in Figure \ref{fig:un_u2_comparison}, showing  $\mb{u}^2$ and $\mb{u}^n$ for a circle network. As a warm-up exercise, let us discuss a limit case. Suppose that $\beta$ is positive but quite small, so that, by (\ref{eq:Neumann}) in Fact \ref{fact:Nash}, $\bm{a}^* \approx \bm{f}$. Then studying statistics of the favorite points $\bm{f}$ is the same as studying the corresponding statistics of $\bm{a}^*$. Let $\mathcal{F}$ be the set of vectors $\bm{f}$ satisfying $\sum_{i \in \mathcal{N}}=0$ and $\Vert \bm{f} \Vert = 1$. We will now note that the extreme eigenvector $\bm{u}^n$ achieves extreme levels of both covariance between neighbors and disagreement disutility. For intuition, consider Figure \ref{fig:un_u2_comparison}.  It is clear that under $\bm{f}=\mb{u}^n$, each agent has a favorite point that is the opposite of those of its neighbors. It is then intuitive that neighbor covariance is as negative as possible: each person disagrees with a random neighbor for sure. Because the costs of disagreement are convex, it is also intuitive that this configuration creates maximum disutility from miscoordination (relative to one where neighbors were closer to each other, as in $\bm{u}^2$). Indeed, by making $\bm{f}$ ``vary gradually'' (changing as little as possible between connected nodes), as in $\bm{u}^2$, we achieve the opposite effect and minimize both disutility and covariance.

These effects are intuitive. However, they do not exhaust the story: to understand how much disutility players experience, we must understand their actions in equilibrium. And, as we have already remarked in presenting Fact \ref{fact:Nash} and in Figure \ref{fig:example-eqm}, for $\beta$ not too close to $0$, these involve substantial attenuation relative to favorite actions. We now turn to explaining this aspect of the result.

Using (\ref{eqn:nashEqm}) and rewriting the condition in the principal component basis, we have $$ \underline{a}_\ell^* = \frac{1-\beta}{1-\beta \lambda_\ell} f_\ell.$$ When only one principal component is represented in the favorite actions, as when $\bm{f}=\bm{u}^n$ or $\bm{f}=\bm{u}^2$, the same is true for equilibrium actions. In other words, in these cases $\bm{a}^*$ is a \emph{scaling} of $\bm{f}$. But the scaling is nontrivial: in best-responding to each other, the disagreement in favorite points is attenuated to a smaller disagreement in equilibrium actions. Indeed, because players best-respond to their neighbors, under $\bm{f}=\mb{u}^n$ they have a strong reason to bring actions closer to zero, in order to coordinate with neighbors. 

Thus, the result involves both forces described above having to do with the structure of $\bm{f}$ alone (which are present even in the $\beta \approx 0$ case) as well as the equilibrium attenuation of $\bm{a}$ (which has a substantial effect only when $\beta $ is far from zero); these forces may pull in opposite directions.

Our result shows that attenuation is \emph{not} enough to overcome the harm done by the strong local disagreement induced by $\bm{u}^n$. One reason for this is that even when players benefit from attenuation by miscoordinating less with neighbors, under $\bm{f}=\bm{u}^n$ they also suffer by being farther from ideal points. It turns out that the planner maximizes their pain by making near neighbors disagree strongly. This pattern, presented in an exteremely simple way for the circle, generalizes to more complex networks as shown in Figure \ref{fig:un_u2_realnetwork}.

On the other hand, a planner who is concerned with creating global disagreement (i.e., minimizing the covariance of a random pair for equilibrium actions) is not at all concerned with making neighbors disagree. For this planner, minimizing attenuation turns out to be the dominant consideration: the planner wants to make sure that as much of the ``size'' of initial disagreement remains in the final equilibrium actions. It is intuitive that this is accomplished by making neighbors \emph{agree} as often as possible. Then strategic forces will not lead them to moderate their behavior by much relative to $\bm{f}$. Of course, the requirement (imposed by definition of $\mathcal{F}$) that $\bm{f}$ have a positive norm, along with the normalization that the average of $\bm{f}$ is equal to zero,  requires heterogeneity across society in favorite points. The best way for a planner to place this heterogeneity is to put the polarization along a ``cut'' such as that depicted in the vectors $\bm{u}^2$ of Figure \ref{fig:un_u2_comparison}. Here disagreement is designed to be as small as possible across most links, and at the optimum, $\bm{f}$ (and, consequently $\bm{a}^*$) will be quite similar for most nodes at short distances. As we have already noted, the configuration $\bm{u}^2$ finds cohesive areas in the network and keeps their $\bm{f}$ similar, while  making relatively ``faraway'' regions disagree with each other. Especially in networks that have good cuts, with large groups that interact fairly little, this is natural: if the global disagreement in $\bm{f}$ is experienced across few links, then it makes little difference to welfare. The vector $\bm{u}^2$ can be seen in a network more interesting than the circle in Figure \ref{fig:un_u2_realnetwork}.

We have spoken informally of $\bm{u}^n$ tending to make neighbors take opposite signs, whereas $\bm{u}^2$ divides the network into cohesive regions. These notions have been extensively formalized in the graph theory literature: see \citet*{DR94}, \citet*{AK97}, and \citet*{Urschel18} for some examples.

\section{Generalizations: General Initial Conditions and Cost Functions}

We return to the general case of the planner's problem stated in (\ref{eqn:genPlanner}): 
\begin{equation*}
    \begin{aligned}
    \max_{\bm{\delta}} \quad & \gamma V(\bm{a}^*)\\
    \textrm{s.t.} \quad & \mb{f} = \hat{\mb{f}} + \bm{\delta}\\
    & \mb{a^*} = (1-\beta)(\mb{I} - \beta\mb{G})^{-1} \mb{f},\\
    & c(\bm{\delta}) \leq C.
\end{aligned}
\end{equation*}

The previous section showed that for very simple planner's constraints, there is a simple description of the most welfare-consequential interventions. However, we worked under many simplifying assumptions: $\hat{\mb{f}}$ was taken to be $\bm{0}$, and the constraint on interventions was to choose one in a ball or on a sphere. 

It is worthwhile to relax both restrictions: we want to consider a status quo that is more flexible. We want to understand to what extent the intuitions extend to more general cost functions.  In this section, we address these issues.

To state results, we need to make a definition measuring the similarity of various vectors to principal components of the underlying network. For this, we use the notion of \emph{cosine similarity}. 

\begin{definition}[Cosine Similarity]
    The \emph{cosine similarity} of two nonzero vectors $\mb{y}$ and $\mb{z}$ is $$\rho (\mb{y}, \mb{z}) = {\mb{y} \cdot \mb{z} \over \left\Vert \mb{y} \right\Vert \left\Vert \mb{z} \right\Vert}.$$
\end{definition}

A canonical interpretation of cosine similarity is that it gives the cosine of the angle between the vectors $\mb{y}$ and $\mb{z}$ in the plane determined by $\mb{y}$ and $\mb{z}$. When $\rho(\mb{y}, \mb{z} ) = 1$ (resp., $-1$), the vector $\mb{z}$ is a positive (resp., negative) rescaling of $\mb{y}$. A cosine similarity of 0 implies that $\mb{y}$ is orthogonal to $\mb{z}$.

\subsection{A monotonicity result}

We are now ready to characterize optimal interventions for a quadratic planner's adjustment cost and arbitrary status quo vector. 

Recall the earlier finding that in the simple  planner's problem with $\gamma=-1$ (malevolent planner) and a constraint of the form $\Vert \bm{f} \Vert \leq 1$, the planner focused \emph{only} on the lowest principal component.  The substance of the next result is that in a suitable sense, this finding generalizes: the planner intervenes more on the principal components with lower eigenvalues. 

\begin{theorem}[Characterization of Optimal Interventions] \label{thm:main}
    Suppose\footnote{Note that we can accommodate any scaling of such a function by suitably adjusting $C$.} $c(\bm{\delta}) =  \Vert \bm{\delta} \Vert^2$. Also suppose that either $\gamma=-1$ or $C$ is small enough that $W(\bm{a}^*)=\bm{0}$ is not feasible for the planner. For generic $\hat{\bm{f}}$, the similarity between $\bm{\delta}^* $ and principal component $\mb{u}^{\ell} (\mb{G})$ satisfies, for $\ell \geq 2$,
    $$\rho (\bm{\delta}^* , \mb{u}^{\ell} ) = \rho (\hat{\mb{f}}, \mb{u}^{\ell} ) \cdot m (\lambda_{\ell}),$$
    where the \textbf{multiplier function} $m$ is such that $|m(\lambda)|$ is decreasing in $\lambda$.
\end{theorem}
\begin{proof}
  Let $\mb{f}^*$ give the optimal choice of $\mb{f}$, so that ${\bm{\delta}}^* = {\mb{f}}^* - \hat{\mb{f}}$. Define $$x_{\ell} = {\underline{f}_{\ell} - \hat{\underline{f}}_{\ell} \over \hat{\underline{f}}_{\ell}}.$$
    
    Then we can rewrite the optimization problem in the principal component basis as follows, for an increasing, negative function $\zeta(\lambda)$:
    \begin{equation} \label{eq:xoptimization}
    \begin{aligned}
        \max_{\mb{x}} \quad & \gamma \sum_{\ell} \zeta(\lambda_{\ell}) (1 + x_{\ell})^2 \hat{\underline{f}}_{\ell}^2\\
        \textrm{s.t.} \quad &  \sum_{\ell} \hat{\underline{f}}_{\ell}^2 x_{\ell}^2 \leq C.\\
    \end{aligned}
    \end{equation}

By our assumption that either $\gamma=-1$ or achieving no miscoordination is infeasible, the budget constraint binds. Thus, letting $\mu$ be the Lagrange multiplier on the budget constraint, the Karush-Kuhn-Tucker necessary condition for optimization is
    $$2 \gamma \hat{\underline{f}}_{\ell}^2 \cdot \zeta(\lambda_{\ell}) (1 + x_{\ell}^*) = 2 \hat{\underline{f}}_{\ell}^2 \cdot \mu  x^*_{\ell}.$$
    Solving for $x_\ell^*$, we get  $\gamma \zeta(\lambda_{\ell})   = x_{\ell}^* (\mu   + \gamma \zeta(\lambda_{\ell}))$, and since the left-hand side is clearly nonzero whenever $\lambda_\ell\neq 1$, it follows that the right-hand side is nonzero too, and we may write
    \begin{equation} {\gamma \zeta(\lambda_{\ell}) \over \mu   + \gamma \zeta(\lambda_{\ell})} = x_{\ell}^*. \label{eq:xsolution}\end{equation} 
    We note a few facts about the solution. From (\ref{eq:xoptimization}) it follows that the $x_\ell$ are all positive at an optimum if $\gamma=-1$ and all negative at an optimum if $\gamma=1$ (by the same argument as in the proof of Theorem 1 of \citet*{GGG20}).\footnote{The intuition is that at $x_\ell=0$, the marginal returns of increasing any $x_\ell$ are nonzero, while the marginal costs are arbitrarily low.}  Lemma \ref{lem:apdx:funcs} gives us that $\zeta$ is a negative, increasing function of its argument. Thus, the denominator $\mu   + \gamma \zeta(\lambda_{\ell})$ in the solution for $x_\ell^*$ is always positive, and $|x_\ell^*|$ is decreasing in $\lambda_\ell$.
    
    Note that $$x_{\ell}^*= {\left\Vert \bm{\delta}^* \right\Vert \rho \left( \bm{\delta}^* , \mb{u}^{\ell} (\mb{G}) \right) \over \left\Vert \hat{\mb{f}} \right\Vert \rho \left( \hat{\mb{f}} , \mb{u}^{\ell} (\mb{G}) \right)}$$ by definition of cosine similarity, so the previous display (\ref{eq:xsolution}) becomes $${\gamma \zeta(\lambda_{\ell}) \over \mu   + \gamma \zeta(\lambda_{\ell})} = {\left\Vert \bm{\delta}^* \right\Vert \rho \left( \bm{\delta}^* , \mb{u}^{\ell} (\mb{G}) \right) \over \left\Vert \hat{\mb{f}} \right\Vert \rho \left( \hat{\mb{f}} , \mb{u}^{\ell} (\mb{G}) \right)}.$$

Rearranging the previous expression gives  $${\rho \left( \bm{\delta}^* , \mb{u}^{\ell} (\mb{G}) \right) } = \rho \left( \hat{\mb{f}} , \mb{u}^{\ell} (\mb{G}) \right) \cdot {\gamma \zeta(\lambda_{\ell}) \over \mu   + \gamma \zeta(\lambda_{\ell})}  \frac{\left\Vert \hat{\mb{f}} \right\Vert}{\left\Vert \bm{\delta}^* \right\Vert }  .$$ By our earlier remark about the monotonicity of $x_\ell^*$ the claim of the proposition follows.\end{proof}

It is worth remarking on a few features of the key expression    $$\rho (\bm{\delta}^* , \mb{u}^{\ell} ) = \rho (\hat{\mb{f}}, \mb{u}^{\ell} ) \cdot m (\lambda_{\ell}).$$ First, the ``status quo term''  $\rho (\hat{\mb{f}}, \mb{u}^{\ell} )$ reflects that the nature of interventions depends on the status quo. For example, if the planner is benevolent and $\hat{\bm{f}}_\ell$ is zero or nearly zero, then there is very little disagreement in that principal component and thus very little to remove; therefore, the planner will not devote a lot of resources to reducing disagreement in that component. The multiplier term captures that components with lower eigenvalues have a bigger welfare impact, and so a planner will care more about adjusting them. 

Crucially, under the assumptions of the theorem, this is true whether the planner is malevolent or benevolent. In the malevolent case, the intuition is exactly the same as that of Proposition \ref{prop:toyResults}: intensifying disagreement in that component has the greatest impact on the disutility of miscoordination, and the planner will want to take advantage of that to increase this disutility. But, under our assumptions that a benevolent planner cannot reach her bliss point of no misscoordination, the intuition applies in the other direction, too: reducing disagreement in the lowest-eigenvalue component is the most effective use of resources to \emph{reduce} disutility.\footnote{Note that the result in Proposition \ref{prop:toyResults}(2) was about a constraint with a \emph{fixed} amount of disagreement, and thus there is no conflict between that result and this intuition.}

\subsection{General cost functions and small budgets}\label{subsec:cost}
A quadratic adjustment cost is a restrictive assumption. Here we show that we can relax this assumption and obtain a version of our result for small budgets $C$, with a simpler characterization of the multiplier function $m$.

 We first make a few assumptions on the structure of the cost function $c(\cdot)$. 

\begin{assumption}[Properties of the Cost Function] \label{as:c}
    The cost function $c (\cdot)$ satisfies the following assumptions: it is twice differentiable; invariant to permutations of the entries of its argument $\bm{\delta}$; nonnegative on its domain; has the value $c(\mb{0}) = \mb{0}$; and has nonsingular Hessian at $\mb{\delta} = \mb{0}$.
\end{assumption}

Making these assumptions implies by standard arguments the approximation 
$$c(\mb{\delta}) = k \left\Vert \mb{\delta} \right\Vert^2 + o(\Vert\mb{\delta}\Vert^2).$$

\begin{proposition}[Characterization of Small Interventions] \label{prop:small}
    Suppose Assumption \ref{as:c} holds. Then for generic $\hat{\bm{f}}$, the similarity between $\bm{\delta}^* $ and principal component $\mb{u}^{\ell} (\mb{G})$ satisfies, for $\ell \geq 2$,
    $$\rho (\bm{\delta}^* , \mb{u}^{\ell} ) = \rho (\hat{\mb{f}}, \mb{u}^{\ell} ) \cdot m (\lambda_{\ell})$$ where  $$ \lim_{C\rightarrow 0} \frac{m(\lambda_\ell)}{m(\lambda_{\ell'})} = \frac{\zeta(\lambda_\ell)}{\zeta(\lambda_{\ell'})}.$$
\end{proposition}

The result follows immediately from Theorem \ref{thm:main} by the same argument as in \citet*[OA3.3]{GGG20}.

Because we have an explicit form for $\zeta$ in Lemma \ref{lem:apdx:funcs}, this result gives a complete description of the optimal intervention. All the cosine similiarities for an orthonormal basis fully pin down the direction of the intervention, and its magnitude is found by exhausting the budget. 

\subsection{An implication for networks with homophily}

We emphasized in Section \ref{sec:interp} that interventions for global discord are extremely different in their form from those for welfare reasons. We can now sketch an application of this to assess whether an intervention is in fact optimal in a practical setting. Our point will be that the characterization permits some simple insights, building on what is known about the spectral structure of real social networks.

Suppose a planner faces a network such as the one shown in Figure \ref{fig:cleavages}, with a certain value of $\lambda_2$, say $\lambda_2 \geq 0.9$ in a homophilous network.\footnote{See \citet*{golub2012does} for more details.} Because $\zeta(\lambda_2)$ is small for large $\lambda_2$, the proposition immediately implies a bound on the cosine similarity $\rho (\bm{\delta}^* , \mb{u}^{\ell} )$: if $m(\lambda_\ell)$ is small, then $\rho (\bm{\delta}^* , \mb{u}^{\ell} )$ is small irrespective of the value of $\hat{\bm{f}}$, since the $\rho (\hat{\mb{f}}, \mb{u}^{\ell} )$ factor in Proposition \ref{prop:small} is bounded by $1$.

\begin{figure}
    \centering
    \includegraphics[width=0.8\textwidth]{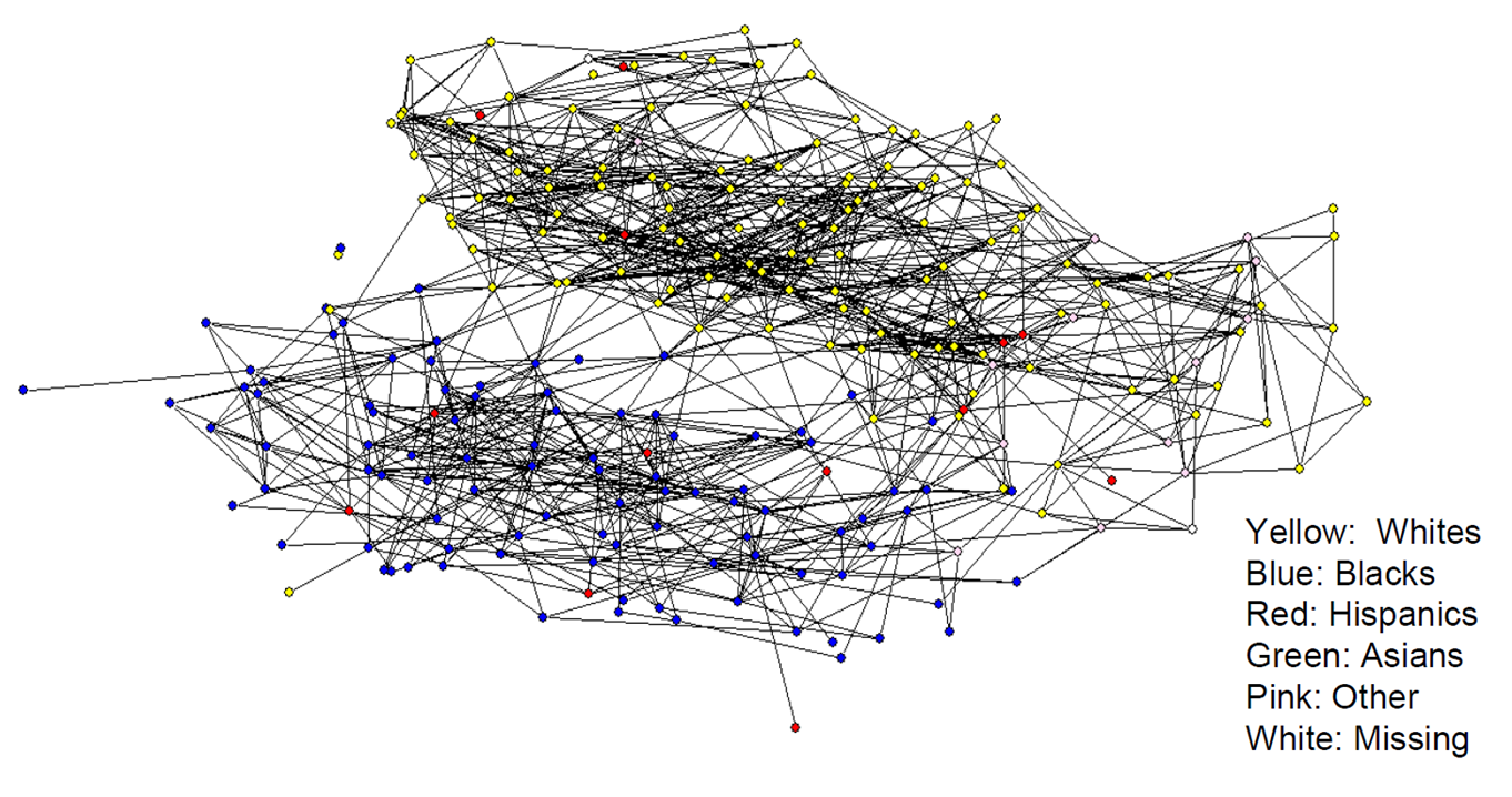}
    \caption{A school network from \citet*{CJP09}.}
    \label{fig:cleavages}
\end{figure}

It follows that if a purportedly optimal intervention has a substantial correlation with $\bm{u}^2$, it is not in fact optimal.\footnote{In practice, $\bm{u}^2$ is highly correlated with demographic covariates (in this example, race), as discussed in \citet*{golub2012does}. So a substantial correlation with race would imply a substantial correlation with $\bm{u}^2$. Thus, one can  refute that an intervention is optimal even without detailed network data, as long as we know that racial homophily is strong.} In other words, welfare-optimal interventions cannot have significant correlation with the main spectral cut of a homophilous network ($\bm{u}^2$). 

\section{Conclusion}

There is a useful duality between the theory of network games and the study of network structure. A familiar pattern goes as follows. We fix a game---e.g., a canonical coordination game---and ask a natural economic question about it, such as what perturbations of agents' ideal points result in large welfare changes. Sometimes, a particular family of network statistics (in this case, the lowest eigenvalue and its associated eigenvector) emerges as an important part of a characterization. Then we have learned both an answer to our economic question and a new interpretation of certain statistics---as well as a new reason to be attentive to the statistics in some situations. 

In this paper, the statistics that emerge from this procedure are $\lambda_n$ and $\bm{u}^n$, as well as other low eigenvalues and eigenvectors. The eigenvalue $\lambda_2$ and the  eigenvector $\bm{u}^2$   have been made famous in both applied mathematics and economics by studies of spectral clustering, homophily, and opinion polarization \citep*{ST07,DVZ03}. But we have spent less time with $\lambda_n$, $\bm{u}^n$, and their friends at the low end of the spectrum. Our analysis here has emphasized their importance for coordination, complementing the findings of some recent studies such as \citet*{BKD14,AK20} and \citet*{GGG20}. More generally, the spectral method for analyzing welfare functionals should be useful for enriching our understanding of the interplay between economic interactions and the networks in which they are embedded.

\bibliographystyle{econ-econometrica}
\bibliography{refs}

\appendix
\section{Functions Used in Spectral Forms of Objectives}
\numberwithin{lemma}{section}

\begin{lemma}\label{lem:apdx:funcs}
    The following functions give the welfare, covariance of neighbors, and covariance of a random pair of agents in the principal component basis. 
    
    \begin{enumerate}
        \item Welfare is given by $$\sum_{\ell=1}^n \zeta (\lambda_{\ell}) {\underline{f}}_{\ell}^2,$$ where $$\zeta ( \lambda) = -\beta (1 - \beta) \frac{(1 - \lambda) [2 - \beta(1 +\lambda)]}{(1 - \beta \lambda)^2}.$$
        \item Covariance of neighbors is given by $$\sum_{\ell=1}^n \eta (\lambda_{\ell}) {\underline{f}}_{\ell}^2,$$ where $$\eta ( \lambda) = {(1 - \beta)^2 \lambda \over (1 - \beta \lambda)^2 n}.$$
        \item Covariance of a random pair is given by $$\sum_{\ell=1}^n \nu (\lambda_{\ell}) {\underline{f}}_{\ell}^2,$$ where $$\nu ( \lambda) = {-(1 - \beta)^2 \over (1 - \beta \lambda)^2 n^2}.$$
    \end{enumerate}
\end{lemma}
\begin{proof}
    The welfare function is given by
    $$V^*=-(1-\beta) \left[\la \mb{f}, \mb{f} \ra + (1- \beta)\la (\mb{I} - \beta \mb{G})^{-1} \mb{f}, ((1+\beta)\mb{I} - 2\beta \mb{G})(\mb{I} - \beta \mb{G})^{-1} \mb{f} - 2 \mb{f} \ra \right].$$
    
    The covariance of neighbors is given by
    $${1 \over n} \la \mb{a}^* , \mb{Ga}^* \ra = {1 \over n} \la (1 - \beta) (\mb{I} - \beta \mb{G})^{-1} \mb{f}, (1 - \beta) \mb{G} (\mb{I} - \beta \mb{G})^{-1} \mb{f} \ra.$$
    
    The covariance of a random pair is given by
    $$-{1 \over n^2} \la \mb{a}^*, \mb{a}^* \ra = -{1 \over n^2} \la (1 - \beta) (\mb{I} - \beta \mb{G})^{-1} \mb{f}, (1 - \beta) (\mb{I} - \beta \mb{G})^{-1} \mb{f} \ra.$$
    
    The $\zeta$, $\eta$, and $\nu$ functions are then immediate by calculation.
\end{proof}

\end{document}